\documentclass{article}
\usepackage{spconf}
\usepackage{cite}
\usepackage{amsfonts, amsmath, amsthm, amscd, amssymb, array, bm, bbm, mathtools, dsfont}
\usepackage{IEEEtrantools}
\usepackage{algorithmic}
\usepackage{graphicx}
\usepackage[title]{appendix}

\DeclareUnicodeCharacter{2212}{-}

\usepackage[bookmarks,colorlinks]{hyperref}

\def\BibTeX{{\rm B\kern-.05em{\sc i\kern-.025em b}\kern-.08em
		T\kern-.1667em\lower.7ex\hbox{E}\kern-.125emX}}

\usepackage[binary-units=true]{siunitx}
\usepackage{pgfplots}
\usepgfplotslibrary{fillbetween}
\usetikzlibrary{arrows,matrix,positioning,shapes.misc,bending,spy}
\pgfplotsset{compat=1.16}
\usepackage{glossaries}
\glsdisablehyper
\newacronym{adc}{ADC}{analog-to-digital converter}
\newacronym{adx}{ADX}{analog-to-digital compression}
\newacronym{awgn}{AWGN}{additive white Gaussian noise}
\newacronym{ct}{CT}{continuous-time}
\newacronym{dmt}{DMT}{discrete multitone transmission}
\newacronym{dt}{DT}{discrete-time}
\newacronym{dtft}{DTFT}{discrete-time \gls{ft}}
\newacronym{ft}{FT}{Fourier transform}
\newacronym{kkt}{KKT}{Karush–Kuhn–Tucker}
\newacronym{lhs}{LHS}{left-hand side}
\newacronym{lmmse}{LMMSE}{linear minimum mean squared error}
\newacronym{mse}{MSE}{mean squared error}
\newacronym{mmse}{MMSE}{minimum mean squared error}
\newacronym{mmwave}{mmWave}{millimeter wave}
\newacronym{mimo}{MIMO}{multiple-input multiple-output}
\newacronym{ofdm}{OFDM}{orthogonal frequency-division multiplexing}
\newacronym{pam}{PAM}{pule-amplitude modulation}
\newacronym{pcm}{PCM}{pulse-code modulation}
\newacronym{pdf}{PDF}{probability density function}
\newacronym{psd}{PSD}{power spectral density}
\newacronym{rhs}{RHS}{right-hand side}
\newacronym{rv}{RV}{random variable}
\newacronym{simo}{SIMO}{single-input multiple-output}
\newacronym{svd}{SVD}{singular value decomposition}
\newacronym{thz}{THz}{terahertz}
\newacronym{tmse}{TMSE}{time-averaged mean squared error}
\newacronym{wss}{WSS}{wide-sense stationary}
\usepackage{xcolor}
\usepackage{url}
\usepackage[capitalise]{cleveref}
\usepackage[shrink=40,letterspace=500]{microtype}
\usepackage{placeins}
\usepackage{appendix}

\usepackage[OT2,T1]{fontenc}
\DeclareSymbolFont{cyrletters}{OT2}{wncyr}{m}{n}
\DeclareMathSymbol{\Sha}{\mathalpha}{cyrletters}{"58}

\DeclareMathOperator*{\argmax}{arg\,max}

\newcommand{\E}{\mathbb{E}}
\newcommand{\FT}[1]{\mathrm{FT}\left\lbrace#1\right\rbrace}
\newcommand{\IFT}[1]{\mathrm{FT}^{-1}\!\left\lbrace#1\right\rbrace}
\newcommand{\DTFT}[1]{\mathrm{DTFT}\left\lbrace#1\right\rbrace}

\newtheorem{lemma}{Lemma}
\newtheorem{proposition}{Proposition}
\newtheorem{theorem}{Theorem}

\newcommand{\iec}{i.e., }

\newcommand{\wrt}{w.r.t.\ }
\allowdisplaybreaks

\title{ON THE ACQUISITION OF STATIONARY SIGNALS USING UNIFORM ADCS}
\name{Peter Neuhaus,
	Nir Shlezinger,
	Meik D\"orpinghaus,
	Yonina C. Eldar,
	and Gerhard Fettweis%
	\thanks{%
		This work has received funding
		from the German Federal Ministry of Education and Research (BMBF) (project E4C, contract number 16ME0189, and project VERITAS, contract number 01IS18073),
		from the European Union’s Horizon 2020 research and innovation program under grant No. 646804-ERC-COG-BNYQ, from the Israel Science Foundation under grant No. 0100101, and from QuantERA grant C’MON-QSENS!.
	}%
	\thanks{
		P. Neuhaus, M. D\"orpinghaus, and G. Fettweis are with the Vodafone Chair Mobile Communications Systems, Technische Universität Dresden, 01062 Dresden, Germany (e-mail: \{peter\_friedrich.neuhaus, meik.doerpinghaus, gerhard.fettweis\}@tu-dresden.de).
		G. Fettweis is also with the Centre for Tactile Internet with Human-in-the-Loop (CeTI) of TU Dresden.\looseness-1
	}%
	\thanks{
		N. Shlezinger is with the School of ECE, Ben-Gurion University of the Negev, Be'er-Sheva 84105, Israel (e-mail: nirshl@bgu.ac.il).
	}%
	\thanks{
		Y. C. Eldar is with the Faculty of Math and CS, Weizmann Institute of Science, Rehovot 7610001, Israel (e-mail:  yonina.eldar@weizmann.ac.il).}%
}
\address{\vspace{-9mm}}
\begin{document}
	\ninept
	\maketitle
	\begin{abstract}
		In this work, we consider the acquisition of stationary signals using uniform analog-to-digital converters (ADCs), i.e., employing uniform sampling and scalar uniform quantization. We jointly optimize the pre-sampling and reconstruction filters to minimize the time-averaged mean-squared error (TMSE) in recovering the continuous-time input signal for a fixed sampling rate and quantizer resolution and obtain closed-form expressions for the minimal achievable TMSE. We show that the TMSE-minimizing pre-sampling filter omits aliasing and discards weak frequency components to resolve the remaining ones with higher resolution when the rate budget is small. In our numerical study, we validate our results and show that sub-Nyquist sampling often minimizes the TMSE under tight rate budgets at the output of the ADC.\looseness-1
	\end{abstract}

	\section{Introduction}\label{sec:intro}
	Analog-to-digital conversion of \gls{ct} processes plays a key role in digital signal processing systems.
	This conversion involves two steps:
	First, the input is sampled, yielding a \gls{dt} representation of the  \gls{ct} signal, and then the samples are mapped onto a finite bit representation, \iec the samples are quantized.
	Such acquisition is typically implemented using uniform \glspl{adc}, which sample at a fixed rate and convert each sample into a digital representation using a uniform partition of the real line. 
	Traditionally, sampling and quantization have been studied independently, with classic results including the Shannon-Nyquist sampling theorem \cite{nyquist1928certain,shannon1949communication} and the 6dB-per-bit rule-of-thumb for high-resolution quantization \cite{bennett1948spectra}.
	A comprehensive overview of works on sampling and quantization can be found in \cite{843002,eldar2015sampling} and \cite{gray1998quantization}, respectively.\looseness-1
	
	While sampling and quantization are often studied separately, \glspl{adc} are typically implemented as part of an overall acquisition system, which also includes analog and digital filters.
	The choice of these filters can contribute to the ability to reconstruct a signal from its digital representation \cite{shlezinger2019joint,644563,1054019,1090615}.
	The \gls{mse}-minimizing pre-filter and recovery filter when considering only sampling and only quantization have been studied in \cite{shlezinger2019joint} and \cite{644563}, respectively.
	The work \cite{1054019} considered both sampling and quantization and numerically optimized the pre-sampling and reconstruction filters with respect to the \gls{mse} distortion in recovering a \gls{ct} \gls{wss} Gaussian process.
	Therein, the optimal pre-sampling filter was only approximated, as it was derived without taking quantization into account, and the resulting \gls{mse} in recovering the input was evaluated numerically.
	In \cite{1090615}, the authors studied a similar setting and analytically characterized the \gls{mse}-minimizing pre-sampling and recovery filters as well as the corresponding minimal achievable \gls{mse}.
	However, none of these works considered both sampling and quantization, while faithfully modeling the quantization distortion.\looseness-1
	
	The fundamental distortion limits in recovering \gls{wss} Gaussian processes from digital representations acquired using filtering, sampling, and quantization have been studied in \cite{8350400,8416707}.
	To characterize the limits, these works allowed the quantization procedure to implement any form of lossy source coding, resulting in a setup denoted as \emph{\gls{adx}}.
	Implementing such acquisition systems involves complex vector quantizers that jointly map an arbitrarily large number of samples onto a discrete representation.
	In addition, the authors of \cite{8350400,8416707} also considered a \gls{pcm} setting, where the acquisition system uses scalar quantizers, which map each sample onto a discrete representation using the same mapping.
	Nonetheless, the effect of quantization in that setting is based on a model which only holds when using fine-resolution non-uniform quantization whose decision regions are tailored to the input distribution.
	Consequently, the resulting model does not reflect the operation of practical acquisition systems, especially when using low-resolution and uniform \glspl{adc}.
	In our previous work \cite{neuhaus2021task}, we circumvented this problem by considering non-subtractive dithered quantization and aimed to recover a random parameter vector which is a linear function of an observed multivariate input process.
	\looseness-1
	
	\begin{figure*}
		\centering
		\includegraphics[width=0.9\textwidth]{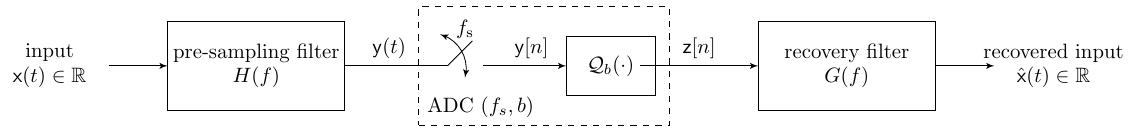}
		\caption{Overview of the system model. Our goal is to recover the process $\mathsf{x}(t)$ under a rate constraint $R = f_\mathrm{s} \cdot b$.}
		\label{fig:system_model}
	\end{figure*}
	
	In contrast, we consider recovering univariate \gls{ct} \gls{wss} signals in this work.
	Particularly, we focus on acquisition systems utilizing conventional uniform \glspl{adc}, and jointly design the pre-sampling filter and the reconstruction filter to minimize the \gls{tmse} in recovering the input when operating under a rate budget at the output of the \gls{adc}.
	To be able to optimize the overall system without imposing distortion models which hold for high-resolution \glspl{adc}, we adopt the approach used in \cite{shlezinger2018hardwarelimited,neuhaus2021task, 8736805} and model the uniform \glspl{adc} as implementing non-subtractive dithered quantization within their dynamic range \cite{gray1993ditheredQuantizers}.
	The resulting model is known to be a faithful approximation of the distortion induced by conventional (including non-dithered) uniform quantizers at arbitrary resolutions for a broad family of input distributions \cite{widrow1996statistical}.   
	Under this model, we analytically characterize the \gls{tmse}-minimizing linear pre-sampling and linear recovery filter as well as the minimal achievable \gls{mse} for a given sampling rate and quantizer resolution.
	Our solution for the pre-sampling filter is shown to be superior to the solution for the \gls{pcm} system given in \cite{8350400,8416707}.
	Moreover, we show numerically that combining the derived filters with conventional, \iec non-dithered quantizers, yields a reduced \gls{mse}, hence, demonstrating the practical value of our contributions.\looseness-1
	
	Throughout this work, random quantities are denoted by sans-serif letters, e.g., $\mathsf{x}$, whereas $x$ is a deterministic quantity.
	We use $j$, $\ast$, $\mathbb{E}\{\cdot\}$, and $\FT{\cdot}$ to denote the imaginary unit, convolution, stochastic expectation, and \gls{ft}, respectively.
	We use $\mathds{1}_{\mathcal{A}(x)}(x)$ to denote the indicator function, which is $1$ when the condition $\mathcal{A}(x)$ holds and $0$ otherwise, and $(x)^+ = \mathrm{max}(0,x)$.
	The sets of natural, integer, real, and complex numbers are written as $\mathbb{N}$, $\mathbb{Z}$, $\mathbb{R}$, and $\mathbb{C}$, respectively.\looseness-1
	
	\section{System Model}\label{sec:system_model}
	\subsection{Acquisition System}\label{sec:acquistion_model}
	We consider the acquisition of a zero-mean \gls{wss} \gls{ct} process $\mathsf{x}(t) \in \mathbb{R}$, $t \in \mathbb{R}$, into a digital representation. 
	We focus on bandlimited inputs, \iec we assume that the \gls{psd} of  $\mathsf{x}(t)$, denoted $S_\mathsf{x}(f)$, is bandlimited with support $\big(-\frac{f_\mathrm{nyq}}{2},\frac{f_\mathrm{nyq}}{2}\big)$, \iec $S_\mathsf{x}(f) = 0$ for all $|f| \geq \frac{f_\mathrm{nyq}}{2}$.
	The conversion system consists of a linear pre-sampling filter $H(f) \in \mathbb{C}$, an \gls{adc} implementing uniform sampling and scalar uniform quantization, and a linear recovery filter $G(f) \in \mathbb{C}$.
	
	The resulting acquisition system, which is illustrated in \cref{fig:system_model}, filters the input by the pre-sampling filter $H(f)$ and subsequently uniformly samples it with sampling rate $f_\mathrm{s}$. The resulting samples are
	\begin{equation}
		\mathsf{y}[n] = \mathsf{y}(n T_\mathrm{s}) = \left( \mathsf{x} * h \right) (n T_\mathrm{s}),\quad  n \in \mathbb{Z},
		\label{eq:def_y_n}
	\end{equation}
	with $h(t) = \IFT{H(f)}$ and $T_\mathrm{s} = \frac{1}{f_\mathrm{s}}$.
	After sampling, $\mathsf{y}[n]$ is quantized by a uniform scalar mid-rise quantizer with an amplitude resolution of $b$ bits, \iec it can produce $2^b$ distinct output values.
	The (one-sided) dynamic range of the quantizer is denoted as $\gamma > 0$, and the mid-rise quantization function is 
	\begin{equation}
		q_b( x' ) = \begin{cases}
			\Delta \left( \left\lfloor \frac{x'}{\Delta} \right\rfloor + \frac{1}{2} \right), & \mathrm{for~} \vert x' \vert < \gamma\\
			\mathrm{sign} \left( x' \right) \left( \gamma - \frac{\Delta}{2} \right), & \mathrm{otherwise,}
		\end{cases}
		\label{eq:def_mid-rise_quantization}
	\end{equation}
	where $\Delta = \frac{2 \gamma}{2^b}$ is the quantization step size, $\lfloor \cdot \rfloor$ denotes rounding to the next smaller integer, and $\mathrm{sign}(\cdot)$ is the signum function. The overall bit rate is thus $R = f_\mathrm{s} \cdot b$ bits per second.
	
	Similar to \cite{shlezinger2018hardwarelimited,neuhaus2021task}, we model the quantizers as implementing \emph{non-subtractive dithered quantization} to obtain an analytically tractable system model.
	Such quantizers add a random \emph{dither} signal to the input before quantization \cite{gray1993ditheredQuantizers}.
	Hence, the quantizer outputs are given by
	\begin{equation}
		\mathsf{z}[n] = \mathcal{Q}_{b}\left( \mathsf{y}[n] \right) = q_b( \mathsf{y}[n] + \mathsf{w}[n] )
		= \mathsf{y}[n] + \mathsf{e}[n].
		\label{eq:def_z_n}
	\end{equation}
	Here, $\mathsf{w}[n]$ denotes the zero-mean dither random process, which is independent and identically distributed (i.i.d.) and mutually independent of the input process, while $\mathsf{e}[n]$ is the quantization distortion.
	For non-overloaded \glspl{adc}, \iec for inputs whose magnitude does not exceed $\gamma$, dithering can ensure that the first and second moments of $\mathsf{e}[n]$ are independent of the input while minimizing the latter by choosing the probability density function of  $\mathsf{w}[n]$ to be a triangular distribution with a width of $2 \Delta$ \cite[Sec. III.C]{wannamaker2000nonsubtractiveDither}.
	Then, to obtain a negligible overload probability,   we set the dynamic range $\gamma$ to a multiple $\eta$ of the standard deviation of the dithered input, \iec\looseness-1
	\begin{equation}
		\gamma^2 =  \eta^2 \, \mathbb{E}\{ (\mathsf{y}[n] + \mathsf{w}[n])^2 \}.
		\label{eq:def_dynamic_range_squared}
	\end{equation}
	By Chebychev's inequality, \eqref{eq:def_dynamic_range_squared} guarantees that the overload probability is not larger than $\eta^{-2}$ for any input distribution \cite[eq.~(5-88)]{papoulis2001probability}.
	The motivation for using the above model stems from the fact that it rigorously yields  a tractable distortion model. In particular, for a vanishing overload probability, it  follows from \cite[Th.~2]{gray1993ditheredQuantizers}	that the autocorrelation function of $\mathsf{e}[n]$ is given by
	$
	R_{\mathsf{e}}[l] = \mathbb{E} \left\lbrace \mathsf{e}[n+l] \mathsf{e}^T[n] \right\rbrace = \frac{\Delta^2}{4} \delta[l],
	\label{eq:def_R_e}
	$
	where $\delta[n]$ denotes the Kronecker delta function.
	While the resulting model of the quantization error rigorously holds for non-overloading non-subtractive dithered quantizers, it also approximately holds for conventional, \iec non-dithered uniform quantizers applied to a broad range of inputs, and particularly sub-Gaussian signals \cite{widrow1996statistical}.
	This is also numerically verified in \cref{sec:numerical_results}.\looseness-1
	
	\subsection{Problem Formulation}\label{sec:problem_formulation}
	We consider the recovery of $\mathsf{x}(t)$ from its digital representation $ \mathsf{z}[n]$.
	To this aim, we focus on \emph{shift-invariant linear recovery} (cf.~\cite{4663942}), \iec a linear recovery filter $G(f)$ is employed, which yields
	\begin{equation}
		\hat{\mathsf{x}}(t) =  \sum_{n \in \mathbb{Z}} g(t - n T_\mathrm{s}) \, \mathsf{z}[n],
		\label{eq:def_xHat}
	\end{equation}
	where $g(t) = \IFT{G(f)}$.
	Our goal is to find the pre-sampling filter $H(f)$ and the  recovery filter $G(f)$, which minimize the reconstruction error for a fixed sampling rate $f_\mathrm{s}$ and quantizer resolution $b$.
	Because $\hat{\mathsf{x}}(t)$ is \emph{cyclostationary} with period $T_\mathrm{s}$ (cf. \cite[Ch.~12]{gardner1986introduction}), we aim to minimize the \gls{tmse} \cite{4663942}, \iec
	\begin{equation}
		\min_{H(f), \, G(f)} \quad \frac{1}{T_\mathrm{s}} \int_{0}^{T_\mathrm{s}} \E \left \lbrace \left \vert \hat{\mathsf{x}}(t) - \mathsf{x}(t) \right \vert^2 \right \rbrace \mathrm{d}t.
		\label{eq:main_objective}
	\end{equation}
	
	Note that while we model and optimize a practical acquisition system architecture in this work, the main focus of the works \cite{8350400,8416707} was to characterize the fundamental distortion limit of \emph{any} acquisition system employing a sampling rate $f_\mathrm{s}$ under the constraint that the samples are encoded with a bit rate $R$.
	Hence, the \gls{adx} results from \cite{8350400,8416707} provide a lower bound on the minimum achievable \gls{tmse} in \eqref{eq:main_objective}.\looseness-1
	
	\section{Uniform ADC Based Acquisition System}\label{sec:main_results}
	Here, we derive the acquisition system which minimizes the \gls{tmse} in recovering $\mathsf{x}(t)$. 
	We obtain the \gls{tmse}-minimizing recovery filter and pre-sampling filter in Subsection~\ref{sec:main_results:Filter}
	and discuss our solution in Subsection~\ref{sec:main_results:discussion}.\looseness-1
	
	\subsection{TMSE Minimizing Filters}\label{sec:main_results:Filter}
	We begin by obtaining the shift-invariant linear recovery filter $G(f)$, which minimizes the \gls{tmse} given in \eqref{eq:main_objective}, for a given pre-sampling filter $H(f)$, a fixed sampling rate $f_\mathrm{s}$, and a fixed quantizer resolution $b$.
	The result is summarized in the following proposition.\looseness-1
	\begin{proposition}\label{prop:opt_recovery_filter}
		For a given pre-sampling filter $H(f)$, the \gls{tmse} minimizing linear recovery filter $G_\mathrm{o}(f)$ is given by
		\begin{equation}
			G_\mathrm{o}(f) = \frac{S_{\mathsf{x}}(f) \, H^*(f)}{ S_\mathsf{y}(e^{j 2 \pi f T_\mathrm{s}}) + \kappa T_\mathrm{s} \int_{-\frac{f_\mathrm{s}}{2}}^{\frac{f_\mathrm{s}}{2}} S_\mathsf{y}(e^{j 2 \pi f' T_\mathrm{s}}) \mathrm{d}f' }
			\label{eq:def_G_opt}
		\end{equation}
		where $S_\mathsf{y}(e^{j 2 \pi f T_\mathrm{s}}) = \frac{1}{T_\mathrm{s}} \sum_{k \in \mathbb{Z}} \left\vert H(f - k f_\mathrm{s}) \right\vert^2 \, S_{\mathsf{x}}(f - k f_\mathrm{s})$ denotes the \gls{psd} of $\mathsf{y}[n]$ and it holds $\kappa = \frac{\eta^2}{2^{2b}} ( 1 - \frac{2 \, \eta^2}{3 \, 2^{2b}} )^{-1}$.
		The resulting minimum achievable \gls{tmse}
		is given by
		\begin{equation}
			\mathrm{TMSE}(H(f)) = \int_{\mathbb{R}} \left( S_{\mathsf{x}}(f) - \frac{ \left\vert H(f) \right\vert^2 S^2_{\mathsf{x}}(f) }{ T_\mathrm{s} \, S_\mathsf{z}(e^{j 2 \pi f T_\mathrm{s}}) } \right) \mathrm{d}f,
			\label{eq:mse_prop_1}
		\end{equation}
		where $S_\mathsf{z}(e^{j 2 \pi f T_\mathrm{s}}) = S_\mathsf{y}(e^{j 2 \pi f T_\mathrm{s}}) + \kappa T_\mathrm{s} \int_{-\frac{f_\mathrm{s}}{2}}^{\frac{f_\mathrm{s}}{2}} S_\mathsf{y}(e^{j 2 \pi f' T_\mathrm{s}}) \mathrm{d}f'$ denotes the \gls{psd} of $\mathsf{z}[n]$.
		
	\end{proposition}
	\begin{proof}
		The proof is provided in Appendix~\ref{sec:appendix_A}.
	\end{proof}
	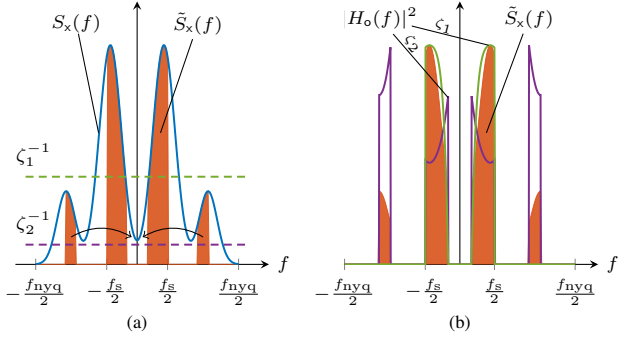
\begin{figure}[t!]
		\centering
		\begin{minipage}[t]{0.49\columnwidth}
			\centering
\definecolor{mycolor1}{rgb}{0.0000,0.4470,0.7410}%
\definecolor{mycolor2}{rgb}{0.8500,0.3250,0.0980}%
\definecolor{mycolor3}{rgb}{0.9290,0.6940,0.1250}%
\definecolor{mycolor4}{rgb}{0.4940,0.1840,0.5560}%
\definecolor{mycolor5}{rgb}{0.4660,0.6740,0.1880}%
\definecolor{mycolor6}{rgb}{0.3010,0.7450,0.9330}%
\definecolor{mycolor7}{rgb}{0.6350,0.0780,0.1840}%

\pgfplotsset{
    every axis/.append style={
        extra description/.code={
            \node at (0.48,-0.225) {(a)};
        },
    },
}

\begin{tikzpicture}[font=\scriptsize,spy using outlines=
	{circle, magnification=5, connect spies}]
	\begin{axis}[%
	width=1.95in,
    height=2in,
    ytick=\empty,
    xtick={-0.5,-0.15,0.15,0.5},
    xticklabels={$-\frac{f_\mathrm{nyq}}{2}$,$-\frac{f_\mathrm{s}}{2}$,$\frac{f_\mathrm{s}}{2}$,$\frac{f_\mathrm{nyq}}{2}$},
    axis x line=bottom,
    axis y line=left,
    axis y line=center,
	xmin=-0.6,
	xmax=0.65,
	xlabel={$f$},
	x label style={at={(axis description cs:1,0)},anchor=west},
	ymin=0,
	ymax=1.2,
	legend columns=1,
	legend style={at={(1,1)},anchor=north east,font=\tiny},
	legend style={legend cell align=left, align=left, draw=white!15!black}
	]
		
		\addplot [name path=tilde_S_x, color=mycolor2, line width=0.1pt,  opacity=0.9]
		table[x index=0,y index=2]{./data/data_illu_Sx_tildeSx_Sy.dat};
		
		\addplot [color=mycolor1, line width=0.75pt]
		table[x index=0,y index=1]{./data/data_illu_Sx_tildeSx_Sy.dat};
        
        \path[name path=xAxis] (axis cs:-0.5,0) -- (axis cs:0.5,0);

        \addplot [
            thick,
            color=mycolor2,
            fill=mycolor2, 
            fill opacity=0.9
        ]
        fill between[
            of=tilde_S_x and xAxis,
            soft clip={domain=-0.5:0.5},
        ];
        
        \addplot [mycolor4, densely dashed, domain=-0.55:0.55, line width=0.75pt] {0.09};
        \addplot [mycolor5, densely dashed, domain=-0.55:0.55, line width=0.75pt] {0.4};
        \node (threshL) at (axis cs:-0.51,0.5){$\zeta^{-1}_{1}$};
        \node (threshH) at (axis cs:-0.51,0.19){$\zeta^{-1}_{2}$};
		
		
		\node (source1) at (axis cs:-0.3,1.1){$S_\mathsf{x}(f)$};
		\coordinate (destination1) at (axis cs:-0.19,0.6){};
		\draw[-,shorten <= -1mm](source1)--(destination1);
		
		\node (source2) at (axis cs:0.3,1.1){$\tilde{S}_\mathsf{x}(f)$};
		\coordinate (destination2) at (axis cs:0.115,0.6){};
		\draw[-,shorten <= -0.5mm](source2)--(destination2);
		
		
		\coordinate (startLeft) at (axis cs:-0.325,0.125);
		\coordinate (endLeft) at (axis cs:-0.0275,0.125);
		\draw[->] (startLeft) to [out=30,in=150] (endLeft);
		\coordinate (startRight) at (axis cs:0.325,0.125);
		\coordinate (endRight) at (axis cs:0.0275,0.125);
		\draw[->] (startRight) to [out=150,in=30] (endRight);
		
	\end{axis}
	
\end{tikzpicture}%
		\end{minipage}%
		\begin{minipage}[t]{0.49\columnwidth}
			\centering
\definecolor{mycolor1}{rgb}{0.0000,0.4470,0.7410}%
\definecolor{mycolor2}{rgb}{0.8500,0.3250,0.0980}%
\definecolor{mycolor3}{rgb}{0.9290,0.6940,0.1250}%
\definecolor{mycolor4}{rgb}{0.4940,0.1840,0.5560}%
\definecolor{mycolor5}{rgb}{0.4660,0.6740,0.1880}%
\definecolor{mycolor6}{rgb}{0.3010,0.7450,0.9330}%
\definecolor{mycolor7}{rgb}{0.6350,0.0780,0.1840}%

\pgfplotsset{
    every axis/.append style={
        extra description/.code={
            \node at (0.455,-0.225) {(b)};
        },
    },
}

\begin{tikzpicture}[font=\scriptsize,spy using outlines=
	{circle, magnification=5, connect spies}]
	\begin{axis}[%
	width=1.95in,
    height=2in,
    ytick=\empty,
    xtick={-0.5,-0.15,0.15,0.5},
    xticklabels={$-\frac{f_\mathrm{nyq}}{2}$,$-\frac{f_\mathrm{s}}{2}$,$\frac{f_\mathrm{s}}{2}$,$\frac{f_\mathrm{nyq}}{2}$},
    axis x line=bottom,
    axis y line=left,
    axis y line=center,
	xmin=-0.5,
	xmax=0.6,
	xlabel={$f$},
	x label style={at={(axis description cs:1,0)},anchor=west},
	ymin=0,
	ymax=1.2,
	legend columns=1,
	legend style={at={(1,1)},anchor=north east,font=\tiny},
	legend style={legend cell align=left, align=left, draw=white!15!black}
	]
		
		
		\addplot [name path=tilde_S_x, color=mycolor2, line width=0.1pt,  opacity=0.9]
		table[x index=0,y index=2]{./data/data_illu_Sx_tildeSx_Sy.dat};
		
		\path[name path=xAxis] (axis cs:-0.5,0) -- (axis cs:0.5,0);

        \addplot [
            thick,
            color=mycolor2,
            fill=mycolor2, 
            fill opacity=0.9
        ]
        fill between[
            of=tilde_S_x and xAxis,
            soft clip={domain=-0.5:0.5},
        ];
        
		\addplot [color=mycolor4, line width=0.75pt]
		table[x index=0,y index=1]{./data/data_illu_H2_opt_highR_lowR.dat};
		
		\addplot [color=mycolor5, line width=0.75pt]
		table[x index=0,y index=2]{./data/data_illu_H2_opt_highR_lowR.dat};
		
		\node (source1) at (axis cs:-0.35,1.125){$|H_\mathsf{o}(f)|^2$};
		\coordinate (destination1) at (axis cs:-0.05,0.76){};
		\coordinate (destination2) at (axis cs:0.125,1){};
		\draw[-,shorten <= -0.5mm](source1)--(destination1) node[pos=0.2,sloped,above=-2.5pt]{\tiny$ \zeta_2$};
		\draw[-,shorten <= -2mm](source1)--(destination2) node[pos=0.25,sloped,above=-2.5pt]{\tiny$ \zeta_1$};;
		
		\node (source3) at (axis cs:0.3,1.125){$\tilde{S}_\mathsf{x}(f)$};
		\coordinate (destination3) at (axis cs:0.115,0.6){};
		\draw[-,shorten <= -0.5mm](source3)--(destination3);
		
	\end{axis}
	
\end{tikzpicture}%
		\end{minipage}%
		\vspace{-3mm}
		\caption{Illustration of the operation of $|H_\mathrm{o}(f)|^2$, where $\zeta_1$ and $\zeta_2$ correspond to a very low and a moderate \gls{adc} amplitude resolution.\vspace{-3mm}}
		\label{fig:illu_H} 
	\end{figure}
	
	Our next goal is to find the \gls{tmse}-minimizing pre-sampling filter, denoted as $H_\mathrm{o}(f)$, by minimizing the \gls{tmse} expression given in \eqref{eq:mse_prop_1} \wrt $H(f)$.
	The result is summarized in the following theorem.\looseness-1
	
	\begin{theorem}\label{theo:optimal_pre-sampling_filter}
		For a fixed sampling rate $f_\mathrm{s}$, and  quantizer resolution $b$, the \gls{tmse} minimizing pre-sampling filter $H_\mathrm{o}(f)$ is characterized by
		\begin{equation}
			\left\vert H_\mathrm{o}\!\left(f - k f_\mathrm{s}\right) \right\vert^2 \!=\!
			\begin{cases}
				\frac{\left(\! \sqrt{ \zeta \, \tilde{S}_{\mathsf{x}}(f) } - 1 \!\right)^{\!+}}{2^{2b} \, \tilde{S}_{\mathsf{x}}(f)}, & k=\tilde{k}(f), \text{~} |f|<\frac{f_\mathrm{s}}{2} \\
				0 & {\rm otherwise},
			\end{cases}
			\label{eq:theorem_def_H_tilde_opt}
		\end{equation}
		with $\tilde{k}(f) = \argmax_{k \in \mathbb{Z}} S_{\mathsf{x}}(f - k f_\mathrm{s})$, $\tilde{S}_{\mathsf{x}}(f) = S_{\mathsf{x}}(f - \tilde{k}(f) f_\mathrm{s}) $, and  $\zeta$ chosen such that 
		$\kappa T_\mathrm{s} \int_{-\frac{f_\mathrm{s}}{2}}^{\frac{f_\mathrm{s}}{2}}\Big( \sqrt{ \zeta \, \tilde{S}_{\mathsf{x}}(f) } - 1 \Big)^+ \mathrm{d}f = 1$.
		Furthermore, the resulting minimum achievable \gls{tmse} is
		\begin{equation}
			\label{eq:theorem_def_TMSE_opt}
			\mathrm{TMSE}_\mathrm{o} \!=\!  \int_{\mathbb{R}} \! S_{\mathsf{x}}(f) \mathrm{d}f \!-\! \int_{-\frac{f_\mathrm{s}}{2}}^{\frac{f_\mathrm{s}}{2}} \! \frac{ \Big(\! \sqrt{ \zeta \, \tilde{S}_{\mathsf{x}}(f) } - 1 \!\Big)^{\!+}\! \tilde{S}_{\mathsf{x}}(f) }{ \Big(\! \sqrt{ \zeta \, \tilde{S}_{\mathsf{x}}(f) } - 1 \!\Big)^{\!+}\! + 1 } \mathrm{d}f.
		\end{equation}
	\end{theorem}
	\begin{proof}
		The proof is provided in Appendix~\ref{sec:appendix_B}.
	\end{proof}
	
	Theorem~\ref{theo:optimal_pre-sampling_filter} characterizes the operation of the \gls{tmse}-minimizing pre-sampling filter $H_\mathrm{o}(f)$.
	Note that the phase of $H_\mathrm{o}(f)$ can be chosen arbitrarily, whereas the phase of $G_\mathrm{o}(f)$ depends on $H(f)$.
	The optimal choice of the pre-sampling filter, i.e., $H_\mathrm{o}(f)$, accounts for both aliasing induced by sub-Nyquist sampling as well as distortion due to low-resolution quantization:
	In particular, the filter preserves only the most dominant spectral components aliased to each frequency after uniform sampling with rate ${f_\mathrm{s}}$ (via the parameter $\tilde{k}(f)$).
	Furthermore, it nullifies the weak spectral modes which are likely to be indistinguishable after uniform quantization (via the parameter $\zeta$).
	
	The operation of the analog filter is illustrated in Fig.~\ref{fig:illu_H}. First, Fig.~\ref{fig:illu_H}.(a) depicts a multi-modal input \gls{psd} $S_\mathsf{x}(f)$, the resulting $\tilde{S}_\mathsf{x}(f)$, and two different water-filling thresholds $\zeta^{-1}_1$ and $\zeta^{-1}_2$, which correspond to a very low and a moderate \gls{adc} amplitude resolution, respectively.
	Then, Fig.~\ref{fig:illu_H}.(b) shows the resulting \gls{tmse}-minimizing pre-sampling filter $|H_\mathrm{o}(f)|^2$ for $\zeta_1$ and $\zeta_2$.
	For a very low \gls{adc} amplitude resolution, i.e., for $\zeta_1$, only the dominant spectral components of $\tilde{S}_\mathsf{x}(f)$ are preserved, whereas in case of a moderate \gls{adc} amplitude resolution, i.e., for $\zeta_2$, all spectral components of $\tilde{S}_\mathsf{x}(f)$ are preserved and  $|H_\mathrm{o}(f)|^2$ converges to a whitening filter.
	
	\subsection{Discussion}\label{sec:main_results:discussion}
	\begin{figure}[t!]
		\centering
\definecolor{mycolor1}{rgb}{0.0000,0.4470,0.7410}%
\definecolor{mycolor2}{rgb}{0.8500,0.3250,0.0980}%
\definecolor{mycolor3}{rgb}{0.9290,0.6940,0.1250}%
\definecolor{mycolor4}{rgb}{0.4940,0.1840,0.5560}%
\definecolor{mycolor5}{rgb}{0.4660,0.6740,0.1880}%
\definecolor{mycolor6}{rgb}{0.3010,0.7450,0.9330}%
\definecolor{mycolor7}{rgb}{0.6350,0.0780,0.1840}%

\begin{tikzpicture}[font=\scriptsize,spy using outlines=
	{circle, magnification=5, connect spies}]
	\begin{axis}[%
	width=2.1in,
    height=2in,
    xtick=\empty,
    ytick=\empty,
    xtick={-0.5,0.5},
    axis x line=bottom,
    axis y line=left,
    axis y line=center,
	xmin=-0.55,
	xmax=0.55,
	xtick={-0.5,0.5},
    xticklabels={$-\frac{f_\mathrm{nyq}}{2}$,$\frac{f_\mathrm{nyq}}{2}$},
	xlabel={$f$},
	x label style={at={(axis description cs:1,0)},anchor=west},
	ymin=0,
	ymax=1.1,
	legend columns=1,
	legend style={at={(1.1,0.5)},anchor=west,font=\scriptsize},
	legend style={legend cell align=left, align=left, draw=white!15!black}
	]
        \addplot[domain=-0.5:0.5, samples=100, color=mycolor1, line width=0.75pt ]{1-2*abs(x)};
        \addlegendentry{$S_\mathsf{x}(f)$}
		
		\addplot[name path=tilde_S_x, color=mycolor2, line width=0.75pt,  opacity=0.9]
		table[x index=0,y index=1]{./data/data_illu_H2_opt_triang_R1_fs_div_fnyq0-9.dat};
		\addlegendentry{$|H_\mathrm{o}(f)|^2$ - \cref{theo:optimal_pre-sampling_filter}}
		
        \addplot[color=mycolor3, line width=0.75pt]
        table[row sep=crcr]{%
            -0.50 0\\
            -0.45 0\\
            -0.45 1\\
            0.45 1\\
            0.45 0\\
            0.5 0\\
        };
        \addlegendentry{$|H_\mathrm{o}(f)|^2$ - PCM \cite{8416707}}
        
		\coordinate (source) at (axis cs:-0.454,0.33){};
		\coordinate (destination) at (axis cs:0.454,0.33){};
		\draw[|-|,line width=0.75pt](source) -- (destination);
		\node (fsLabel) at (axis cs:-0.06,0.38) {\scriptsize $f_\mathrm{s}$};
		
		\coordinate (source2) at (axis cs:-0.379,0.66){};
		\coordinate (destination2) at (axis cs:0.379,0.66){};
		\draw[|-|,line width=0.75pt](source2) -- (destination2);
		\node (fHLabel) at (axis cs:-0.06,0.71) {\scriptsize $f_H$};
		
		\coordinate (vertLeftHigh) at (axis cs:-0.375,0.66){};
		\coordinate (vertLeftLow) at (axis cs:-0.375,0){};
		\draw[densely dotted, line width=0.75pt](vertLeftHigh) -- (vertLeftLow);
		\coordinate (vertRightHigh) at (axis cs:0.375,0.66){};
		\coordinate (vertRightLow) at (axis cs:0.375,0){};
		\draw[densely dotted, line width=0.75pt](vertRightHigh) -- (vertRightLow);
		
	\end{axis}
	
\end{tikzpicture}%
		\vspace{0.3mm}
		\caption{Comparison of $|H_\mathrm{o}(f)|^2$ from \cref{theo:optimal_pre-sampling_filter} to the filter of \cite{8416707} for \hbox{$R=1$}~bit per Nyquist interval.\vspace{-5mm}}
		\label{fig:illu_cmp_H2_opt} 
	\end{figure}
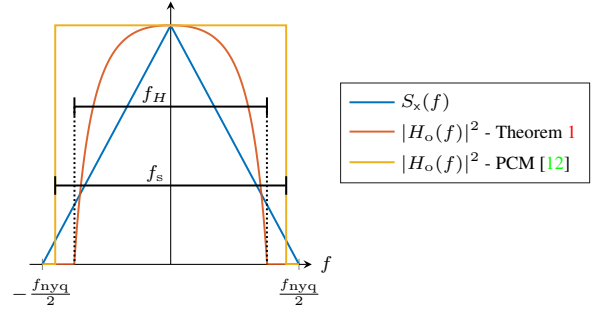
	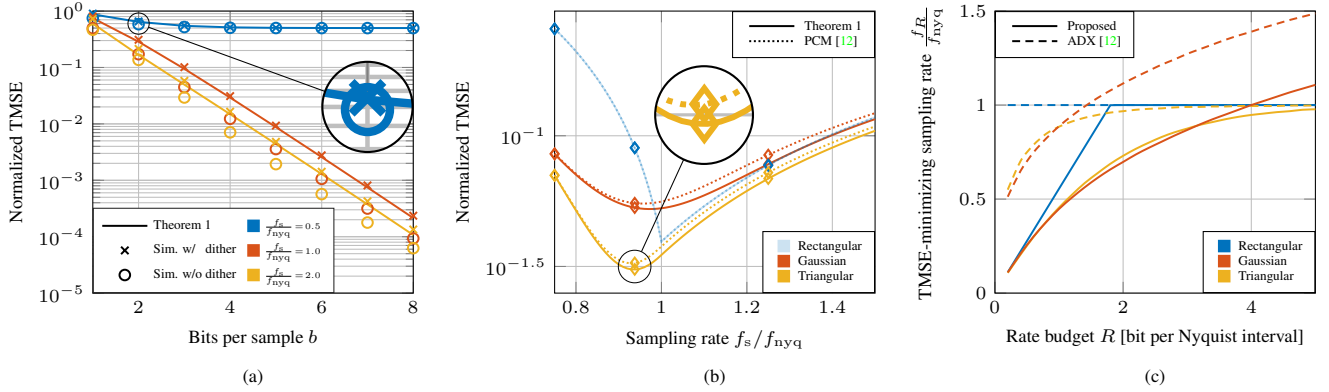
\begin{figure*}
		\centering
		\begin{minipage}[t]{0.33\textwidth}
			\centering
\definecolor{mycolor1}{rgb}{0.0000,0.4470,0.7410}%
\definecolor{mycolor2}{rgb}{0.8500,0.3250,0.0980}%
\definecolor{mycolor3}{rgb}{0.9290,0.6940,0.1250}%
\definecolor{mycolor4}{rgb}{0.4940,0.1840,0.5560}%
\definecolor{mycolor5}{rgb}{0.4660,0.6740,0.1880}%
\definecolor{mycolor6}{rgb}{0.3010,0.7450,0.9330}%
\definecolor{mycolor7}{rgb}{0.6350,0.0780,0.1840}%

\pgfplotsset{
    every axis/.append style={
        extra description/.code={
            \node at (0.5,-0.3) {(a)};
        },
    },
}

\begin{tikzpicture}[font=\scriptsize,spy using outlines=
	{circle, magnification=4.25, connect spies}]
	\begin{axis}[%
	width=2.3in,
    height=2.1in,
	xmin=1,
	xmax=8,
	xlabel={Bits per sample $b$},
	ymode=log,
	ymin=1e-5,
	ymax=1,
	yminorticks=true,
	ylabel={Normalized TMSE},
	xmajorgrids,
	ymajorgrids,
	yminorgrids,
	ylabel near ticks,
    xlabel near ticks,
	legend columns=2,
	legend style={at={(0,0)},anchor=south west,font=\tiny},
	legend style={legend cell align=left, align=left, draw=white!15!black}
	]

		\addlegendimage{black, line width=0.75pt}
		\addlegendentry{Theorem 1}
		\addlegendimage{color=mycolor1, mark options={solid,fill=mycolor1}, only marks, mark=square*, line width=0.75pt}
		\addlegendentry{\scalebox{0.65}{$\frac{f_\mathrm{s}}{f_\mathrm{nyq}} \!=\! 0.5$}}
		\addlegendimage{black,  mark=x, mark options={solid, black}, only marks, line width=0.75pt}
		\addlegendentry{Sim. w/~~~dither~~}
		\addlegendimage{color=mycolor2, mark options={solid,fill=mycolor2}, only marks, mark=square*, line width=0.75pt}
		\addlegendentry{\scalebox{0.65}{$\frac{f_\mathrm{s}}{f_\mathrm{nyq}} \!=\! 1.0$}}
		\addlegendimage{black,  mark=o, mark options={solid, black}, only marks, line width=0.75pt}
		\addlegendentry{Sim. w/o dither~~}
		\addlegendimage{color=mycolor3, mark options={solid,fill=mycolor3}, only marks, mark=square*, line width=0.75pt}
		\addlegendentry{\scalebox{0.65}{$\frac{f_\mathrm{s}}{f_\mathrm{nyq}} \!=\! 2.0$}}

		\addplot [color=mycolor1, line width=0.75pt]
		table[x index=0,y index=1]{./data/eval_model_bit-theo-sim_dithered_fs-div-fNyq_0-5.dat};
		\addplot [color=mycolor1, mark=x, mark options={solid, mycolor1}, only marks, line width=0.75pt]
		table[x index=0,y index=2]{./data/eval_model_bit-theo-sim_dithered_fs-div-fNyq_0-5.dat};
		\addplot [color=mycolor1, mark=o, mark options={solid, mycolor1}, only marks, line width=0.75pt]
		table[x index=0,y index=2]{./data/eval_model_bit-theo-sim_non-dithered_fs-div-fNyq_0-5.dat};
		
		\addplot [color=mycolor2, line width=0.75pt]
		table[x index=0,y index=1]{./data/eval_model_bit-theo-sim_dithered_fs-div-fNyq_1.dat};
		\addplot [color=mycolor2, mark=x, mark options={solid, mycolor2}, only marks, line width=0.75pt]
		table[x index=0,y index=2]{./data/eval_model_bit-theo-sim_dithered_fs-div-fNyq_1.dat};
		\addplot [color=mycolor2, mark=o, mark options={solid, mycolor2}, only marks, line width=0.75pt]
		table[x index=0,y index=2]{./data/eval_model_bit-theo-sim_non-dithered_fs-div-fNyq_1.dat};
		
		\addplot [color=mycolor3, line width=0.75pt]
		table[x index=0,y index=1]{./data/eval_model_bit-theo-sim_dithered_fs-div-fNyq_2.dat};
		\addplot [color=mycolor3, mark=x, mark options={solid, mycolor3}, only marks, line width=0.75pt]
		table[x index=0,y index=2]{./data/eval_model_bit-theo-sim_dithered_fs-div-fNyq_2.dat};
		\addplot [color=mycolor3, mark=o, mark options={solid, mycolor3}, only marks, line width=0.75pt]
		table[x index=0,y index=2]{./data/eval_model_bit-theo-sim_non-dithered_fs-div-fNyq_2.dat};
		
		\coordinate (spypoint) at (axis cs:2,0.6);
        \coordinate (magnifyglass) at (axis cs:7,0.02);
		
	\end{axis}

    \spy [black, size=1.2cm] on (spypoint)
        in node[fill=white] at (magnifyglass);

\end{tikzpicture}%
		\end{minipage}%
		\hspace{0.8mm}%
		\begin{minipage}[t]{0.33\textwidth}
			\centering
\definecolor{mycolor1}{rgb}{0.0000,0.4470,0.7410}%
\definecolor{mycolor3}{rgb}{0.8500,0.3250,0.0980}%
\definecolor{mycolor2}{rgb}{0.9290,0.6940,0.1250}%
\definecolor{mycolor4}{rgb}{0.4940,0.1840,0.5560}%
\definecolor{mycolor5}{rgb}{0.4660,0.6740,0.1880}%
\definecolor{mycolor6}{rgb}{0.3010,0.7450,0.9330}%
\definecolor{mycolor7}{rgb}{0.6350,0.0780,0.1840}%

\pgfplotsset{
    every axis/.append style={
        extra description/.code={
            \node at (0.5,-0.3) {(b)};
        },
    },
}

\begin{tikzpicture}[font=\scriptsize,spy using outlines=
	{circle, magnification=3, connect spies}]
	\begin{axis}[%
	width=2.3in,
    height=2.1in,
	xmin=0.75,
	xmax=1.5,
	xlabel={Sampling rate $f_\mathrm{s}/f_\mathrm{nyq}$},
	ymode=log,
	ymin=0.025,
	ymax=0.3,
	yminorticks=true,
	ylabel={Normalized TMSE},
	xmajorgrids,
	ymajorgrids,
	yminorgrids,
	ylabel near ticks,
    xlabel near ticks,
	legend columns=1,
	legend style={at={(1,1)},anchor=north east,font=\tiny},
	legend style={legend cell align=left, align=left, draw=white!15!black}
	]

		\addplot[color=black, line width=0.75pt,
        y filter/.code={\pgfmathparse{\pgfmathresult-0}\pgfmathresult}]
          table[row sep=crcr]{%
        	-15 -5\\
        };\label{P11}
        \addplot[color=black, densely dotted, line width=0.75pt,
        y filter/.code={\pgfmathparse{\pgfmathresult-0}\pgfmathresult}]
          table[row sep=crcr]{%
        	-15 -5\\
        };\label{P12}
        \addplot[color=black, mark=star, only marks, line width=0.75pt,
        y filter/.code={\pgfmathparse{\pgfmathresult-0}\pgfmathresult}]
          table[row sep=crcr]{%
        	-15 -5\\
        };\label{P13}
        
        \node [draw,fill=white,font=\tiny,anchor=north east] at (axis cs: 1.5,0.3) {
        \setlength{\tabcolsep}{0.6mm}
        \begin{tabular}{c l}
        \ref{P11} & {Theorem~1}\\
        \ref{P12} & {PCM \cite{8416707}}\\
        \end{tabular}
        };
		
		\addplot[color=mycolor1, mark options={solid,fill=mycolor1}, only marks, mark=square*, line width=0.75pt, opacity=0.2,
        y filter/.code={\pgfmathparse{\pgfmathresult-0}\pgfmathresult}]
          table[row sep=crcr]{%
        	-15 -5\\
        };\label{P21}
        
        \addplot[color=mycolor2, mark options={solid,fill=mycolor2}, only marks, mark=square*, line width=0.75pt,
        y filter/.code={\pgfmathparse{\pgfmathresult-0}\pgfmathresult}]
          table[row sep=crcr]{%
        	-15 -5\\
        };\label{P22}
        
        \addplot[color=mycolor3, mark options={solid,fill=mycolor3}, only marks, mark=square*, line width=0.75pt,
        y filter/.code={\pgfmathparse{\pgfmathresult-0}\pgfmathresult}]
          table[row sep=crcr]{%
        	-15 -5\\
        };\label{P23}
        
        \node [draw,fill=white,font=\tiny,anchor=south east] at (axis cs: 1.5,0.025) {
        \setlength{\tabcolsep}{0.6mm}
        \begin{tabular}{l l}
        \ref{P21} & {Rectangular}\\
        \ref{P23} & {Gaussian}\\
        \ref{P22} & {Triangular}\\
        \end{tabular}
        };

		\addplot [color=mycolor3, line width=0.75pt]
		table[x index=0,y index=1]{./data/cmp_kipnis_gauss_fs-div-fNyq_prop_pcm_adx_R3-75.dat};

		\addplot [color=mycolor3, densely dotted, line width=0.75pt]
		table[x index=0,y index=2]{./data/cmp_kipnis_gauss_fs-div-fNyq_prop_pcm_adx_R3-75.dat};
		
		\addplot [color=mycolor3, line width=0.75pt,  mark=diamond, only marks]
		table[x index=0,y index=1]{./data/cmp_kipnis_gauss_fs-div-fNyq_discrete-b_prop_pcm_R3-75.dat};
		\addplot [color=mycolor3, line width=0.75pt, mark=diamond, only marks]
		table[x index=0,y index=2]{./data/cmp_kipnis_gauss_fs-div-fNyq_discrete-b_prop_pcm_R3-75.dat};
        
		\addplot [color=mycolor1, line width=0.75pt,opacity=0.2]
		table[x index=0,y index=1]{./data/cmp_kipnis_rect_fs-div-fNyq_prop_pcm_adx_R3-75.dat};

		\addplot [color=mycolor1, densely dotted, line width=0.75pt,opacity=0.4]
		table[x index=0,y index=2]{./data/cmp_kipnis_rect_fs-div-fNyq_prop_pcm_adx_R3-75.dat};
		
		\addplot [color=mycolor1, line width=0.75pt,  mark=diamond, only marks]
		table[x index=0,y index=1]{./data/cmp_kipnis_rect_fs-div-fNyq_discrete-b_prop_pcm_R3-75.dat};
		\addplot [color=mycolor1, line width=0.75pt, mark=diamond, only marks]
		table[x index=0,y index=2]{./data/cmp_kipnis_rect_fs-div-fNyq_discrete-b_prop_pcm_R3-75.dat};

		\addplot [color=mycolor2, line width=0.75pt]
		table[x index=0,y index=1]{./data/cmp_kipnis_triang_fs-div-fNyq_prop_pcm_adx_R3-75.dat};

		\addplot [color=mycolor2, densely dotted, line width=0.75pt]
		table[x index=0,y index=2]{./data/cmp_kipnis_triang_fs-div-fNyq_prop_pcm_adx_R3-75.dat};
		
		\addplot [color=mycolor2, line width=0.75pt,  mark=diamond, only marks]
		table[x index=0,y index=1]{./data/cmp_kipnis_triang_fs-div-fNyq_discrete-b_prop_pcm_R3-75.dat};
		\addplot [color=mycolor2, line width=0.75pt, mark=diamond, only marks]
		table[x index=0,y index=2]{./data/cmp_kipnis_triang_fs-div-fNyq_discrete-b_prop_pcm_R3-75.dat};
		
		\coordinate (spypoint) at (axis cs:0.9375,0.0316);
        \coordinate (magnifyglass) at (axis cs:1.1,0.12);
		
	\end{axis}
	
    \spy [black, size=1.3cm] on (spypoint)
        in node[fill=white] at (magnifyglass);

\end{tikzpicture}%
		\end{minipage}%
		\hspace{0.8mm}%
		\begin{minipage}[t]{0.33\textwidth}
			\centering
\definecolor{mycolor1}{rgb}{0.0000,0.4470,0.7410}%
\definecolor{mycolor3}{rgb}{0.8500,0.3250,0.0980}%
\definecolor{mycolor2}{rgb}{0.9290,0.6940,0.1250}%
\definecolor{mycolor4}{rgb}{0.4940,0.1840,0.5560}%
\definecolor{mycolor5}{rgb}{0.4660,0.6740,0.1880}%
\definecolor{mycolor6}{rgb}{0.3010,0.7450,0.9330}%
\definecolor{mycolor7}{rgb}{0.6350,0.0780,0.1840}%

\pgfplotsset{
    every axis/.append style={
        extra description/.code={
            \node at (0.5,-0.3) {(c)};
        },
    },
}

\begin{tikzpicture}[font=\scriptsize,spy using outlines=
	{circle, magnification=5, connect spies}]
	\begin{axis}[%
	width=2.3in,
    height=2.1in,
	xmin=0,
	xmax=5,
	xlabel={Rate budget $R$ [bit per Nyquist interval]},
	ymin=0,
	ymax=1.5,
	yminorticks=true,
	ylabel={TMSE-minimizing sampling rate $\frac{f_R}{f_\mathrm{nyq}}$},
	xmajorgrids,
	ymajorgrids,
	yminorgrids,
	ylabel near ticks,
    xlabel near ticks,
	legend columns=1,
	legend style={at={(1,1)},anchor=north east,font=\tiny},
	legend style={legend cell align=left, align=left, draw=white!15!black}
	]
		\addplot[color=black, line width=0.75pt,
        y filter/.code={\pgfmathparse{\pgfmathresult-0}\pgfmathresult}]
          table[row sep=crcr]{%
        	-15 -5\\
        };\label{c11}
        \addplot[color=black, densely dotted, line width=0.75pt,
        y filter/.code={\pgfmathparse{\pgfmathresult-0}\pgfmathresult}]
          table[row sep=crcr]{%
        	-15 -5\\
        };\label{c12}
        \addplot[color=black, densely dashed, line width=0.75pt,
        y filter/.code={\pgfmathparse{\pgfmathresult-0}\pgfmathresult}]
          table[row sep=crcr]{%
        	-15 -5\\
        };\label{c13}
        
        \node [draw,fill=white,font=\tiny,anchor=north west] at (axis cs: 0,1.5) {
        \setlength{\tabcolsep}{0.6mm}
        \begin{tabular}{c l}
        \ref{c11} & {Proposed}\\
        {\ref{c13}} & {ADX \cite{8416707}}\\
        \end{tabular}
        };

		\addplot[color=mycolor1, mark options={solid,fill=mycolor1}, only marks, mark=square*, line width=0.75pt,
        y filter/.code={\pgfmathparse{\pgfmathresult-0}\pgfmathresult}]
          table[row sep=crcr]{%
        	-15 -5\\
        };\label{b1}
        
        \addplot[color=mycolor2, mark options={solid,fill=mycolor2}, only marks, mark=square*, line width=0.75pt,
        y filter/.code={\pgfmathparse{\pgfmathresult-0}\pgfmathresult}]
          table[row sep=crcr]{%
        	-15 -5\\
        };\label{b2}
        
        \addplot[color=mycolor3, mark options={solid,fill=mycolor3}, only marks, mark=square*, line width=0.75pt,
        y filter/.code={\pgfmathparse{\pgfmathresult-0}\pgfmathresult}]
          table[row sep=crcr]{%
        	-15 -5\\
        };\label{b3}
        
        \node [draw,fill=white,font=\tiny,anchor=south east] at (axis cs: 5,0.0) {
        \setlength{\tabcolsep}{0.6mm}
        \begin{tabular}{l l}
        \ref{b1} & {Rectangular}\\
        \ref{b3} & {Gaussian}\\
        \ref{b2} & {Triangular}\\
        \end{tabular}
        };
        
        \addplot [color=mycolor1, line width=0.75pt]
        table[x index=0,y index=1]{./data/eval_fR_over_R_rect_prop_adx_v2.dat};
        \addplot [color=mycolor1, densely dashed, line width=0.75pt]
        table[x index=0,y index=2]{./data/eval_fR_over_R_rect_prop_adx_v2.dat};
		
		\addplot [smooth, color=mycolor2, line width=0.75pt]
		table[x index=0,y index=1]{./data/eval_fR_over_R_triang_prop_adx_v2.dat};
		\addplot [smooth, color=mycolor2, densely dashed, line width=0.75pt]
		table[x index=0,y index=2]{./data/eval_fR_over_R_triang_prop_adx_v2.dat};
		
		\addplot [smooth, color=mycolor3, line width=0.75pt]
		table[x index=0,y index=1]{./data/eval_fR_over_R_gauss_prop_adx_v2.dat};
		\addplot [smooth, color=mycolor3, densely dashed, line width=0.75pt]
		table[x index=0,y index=2]{./data/eval_fR_over_R_gauss_prop_adx_v2.dat};
		
	\end{axis}
	
\end{tikzpicture}%
		\end{minipage}%
		\vspace{-3mm}
		\caption{In (a), we evaluate the \gls{tmse} over an increasing number of bits per sample $b$ for different sampling rates $f_\mathrm{s}$.
			We also compare the \gls{tmse} predicted by \cref{theo:optimal_pre-sampling_filter} to the simulated \gls{tmse} when employing non-subtractive dithered and conventional, \iec non-dithered, quantizers.
			In (b), we compare the achievable \gls{tmse} when employing our proposed pre-sampling filter to the performance when employing the one proposed in \cite{8416707}, where we modify the result from \cite{8416707} by using the quantization distortion resulting from our system model to enable a fair comparison.
			Diamonds illustrate the sampling rates corresponding to integer values of $b$, \iec $b=R/f_\mathrm{s} \in \mathbb{N}$.
			In (c), we evaluate the sampling rates $f_R$, which minimize the \gls{tmse} for a fixed rate budget $R=f_\mathrm{s} \cdot b$.
			Note that we set the 3-dB bandwidth of the Gaussian \gls{psd} to $\frac{f_\mathrm{nyq}}{2}$ in this work.\looseness-1\vspace{-3mm}}
		\label{fig:num_res} 
	\end{figure*}
	
	Together, \cref{prop:opt_recovery_filter} and \cref{theo:optimal_pre-sampling_filter} characterize the \gls{tmse} minimizing acquisition system, which utilizes a uniform scalar \gls{adc}.
	The analysis presented in this work is different from that provided in the \gls{pcm} setting of \cite[Sec.~V]{8416707}, which is carried out \emph{assuming} that the quantization distortion can be modeled as an additive zero-mean distortion, which is white, i.e., temporally uncorrelated, and uncorrelated with the quantizer input.
	In contrast, we can \emph{guarantee} the same properties for any non-overloading quantizer input distribution by employing non-subtractive dithered quantization, which, however, complicates the derivation.
	
	While our derived recovery filter $G_\mathrm{o}(f)$ and the corresponding \gls{tmse} expression from \cref{prop:opt_recovery_filter} are similar to those obtained in \cite[Prop.~5]{8416707}, i.e., both are Wiener filters and differences are due to different quantization distortion models, the proposed pre-sampling filters $H_\mathrm{o}(f)$ also differ.
	In particular, for symmetric input \glspl{psd} $S_\mathsf{x}(f)$ which are non-increasing for $f>0$, \iec for unimodal \glspl{psd} $S_\mathsf{x}(f)$, the authors of \cite{8416707} assume that \eqref{eq:mse_prop_1} is minimized by choosing $H(f)$ as a low-pass filter with cut-off frequency $\frac{f_\mathrm{s}}{2}$ \cite[Sec.~V.B]{8416707}.
	This means that $H(f)$ corresponds to a conventional anti-aliasing filter.
	In contrast, \cref{theo:optimal_pre-sampling_filter} proves that the \gls{tmse}-minimizing pre-sampling filter, \iec $H_\mathrm{o}(f)$, may discard weak frequency components of $S_\mathsf{x}(f)$, which cannot be resolved with the given quantizer resolution.
	For example, again considering unimodal \glspl{psd} $S_\mathsf{x}(f)$, the frequency support of $H_\mathrm{o}(f)$ may be given by $\mathds{1}_{|f| < \frac{f_H}{2}}(f)$ with $f_H < f_\mathrm{s}$, as illustrated in \cref{fig:illu_cmp_H2_opt}.
	This allows to resolve the remaining frequency components with higher accuracy, because the quantization distortion is proportional to the variance of the samples at the quantizer input, \iec $\mathbb{E}\{\mathsf{e}^2[n]\} \propto \mathbb{E}\{\mathsf{y}^2[n]\}$.
	Consequently, the minimal \gls{tmse} is achieved by a trade-off between the distortion due to the pre-sampling filter and the quantization distortion.
	Note that a similar result has been obtained in \cite{644563}, where the authors only consider quantization.
	In \cref{sec:numerical_results} we verify our claim and show that the pre-sampling filter proposed in \cref{theo:optimal_pre-sampling_filter} can outperform the one considered in \cite{8416707}.\looseness-1
	
	While the pre-sampling filter proposed here is generally different from the one employed in \cite{8416707}, they are identical up to a scaling factor for a rectangular input \gls{psd} $S_\mathsf{x}(f)$, because in this case, there are no weak frequency components to be discarded.
	For such rectangular input \glspl{psd}, one can specialize \eqref{eq:theorem_def_TMSE_opt} to a closed-form expression.
	In the following, we compare this expression to the normalized \gls{tmse} of the solutions from \cite{8416707}, which yields the normalized \gls{tmse}
	\begin{align}
		& \frac{\mathrm{TMSE}_\mathrm{o}(f_\mathrm{s},b)}{\E \left \lbrace \left \vert \mathsf{x}(t) \right \vert^2 \right \rbrace} = 1 - \frac{\min(f_\mathrm{s},f_\mathrm{nyq})}{f_\mathrm{nyq}} \nonumber \\
		& \qquad \times \begin{cases}
			\left(1+T_\mathrm{s} \! \min(f_\mathrm{s},f_\mathrm{nyq}) \, \bar{\kappa} 2^{-2b}\right)^{-1}, & \text{Proposed}\\
			\left(1 + T_\mathrm{s} \! \min(f_\mathrm{s},f_\mathrm{nyq}) \, c_\mathrm{q} 2^{-2b}\right)^{-1}, & \text{PCM~\cite{8416707}}\\
			(1-2^{-2b}), & \text{ADX~\cite{8416707}}
		\end{cases},
		\label{eq:ntmse_rect}
	\end{align}
	with $\bar{\kappa} = \eta^2 ( 1 - \frac{2 \, \eta^2}{3 \, 2^{2b}} )^{-1}$.
	In \eqref{eq:ntmse_rect}, it holds $c_\mathrm{q} = \frac{\sqrt{3}\pi}{2}$ for non-uniform scalar quantization of Gaussian inputs \cite[p.~2329]{gray1998quantization}, as considered in \cite{8416707}.
	Furthermore, the factor `$\min(f_\mathrm{s},f_\mathrm{nyq}) f_\mathrm{nyq}^{-1}$' corresponds to the loss due to sub-Nyquist sampling, `$\mathcal{X} \cdot 2^{−2b}$', $\mathcal{X} \in \{1,\bar{\kappa},c_\mathrm{q}\}$ stems from the loss due to quantization, and the factor `$T_\mathrm{s} \! \min(f_\mathrm{s},f_\mathrm{nyq})$' shows a reduction of the quantization distortion when employing temporal oversampling.
	It can be shown that the \gls{tmse} for \gls{pcm} in \eqref{eq:ntmse_rect} is lower than that of Theorem~\ref{theo:optimal_pre-sampling_filter}, while both are outperformed by \gls{adx}, which corresponds to a lower bound on the minimum achievable distortion for any practical system.
	However, the \gls{pcm} \gls{tmse} expression assumes non-uniform quantization and is only valid for Gaussian input distributions, while the proposed \gls{tmse} is derived for uniform quantization and holds for arbitrary input distributions when employing non-subtractive dithered quantization.\looseness-1
	
	The analysis in this work is limited to single branch sampling.
	In \cite[Sec.~2.5]{kipnis2017fundamental}, it was shown that single branch sampling is sufficient for unimodal input \glspl{psd}.
	However, for multi-modal input \glspl{psd}, the performance can generally be improved by employing multi-branch sampling, which is beyond the scope of this work.
	
	\section{Numerical Results}\label{sec:numerical_results}
	In this section, we provide a numerical study focusing on Gaussian processes $\mathsf{x}(t)$.
	For the considered equivalent quantizer model given in \eqref{eq:def_z_n}, it has been shown in \cite[Sec.~IV.B]{neuhaus2021task} that the overload probability needs to decrease with increasing $b$ to ensure its validity.
	Hence, in the following, we increase $\eta$ with $b$, i.e., we set $\eta(b) = 0.25b + 1.75$, as proposed in \cite[Sec.~IV.B]{neuhaus2021task}.
	
	First, we validate our theoretical results in \cref{fig:num_res}.(a), asserting that the \gls{tmse} derived in Theorem~\ref{theo:optimal_pre-sampling_filter}, where non-overloaded \glspl{adc} are assumed, is indeed achievable by the proposed system with and without dithering.
	Here, we consider a rectangular input \gls{psd} $S_\mathsf{x}(f)$.
	It can be seen that the \gls{tmse} predicted by \eqref{eq:theorem_def_TMSE_opt} is a close match to the simulated \gls{tmse}, when employing dithered quantizers, as we assume in our system model.
	The mismatch is likely due to a small but non-zero overload probability.
	Furthermore, we observe that a \emph{lower} \gls{tmse} can be achieved when using the same filters but employing conventional  non-dithered quantizers.
	This asserts the validity of our derivations.\looseness-1
	
	Next, in \cref{fig:num_res}.(b), we compare the resulting \gls{tmse} when employing the pre-sampling filter proposed in \cref{theo:optimal_pre-sampling_filter} to the one derived in \cite[Sec.~V.B]{8416707}, while using the quantization noise model derived in this work, e.g., we set $c_\mathrm{q}=\bar{\kappa}$ in \eqref{eq:ntmse_rect}.
	Here, we fix the rate budget to $R= f_\mathrm{s} \cdot b = \SI{3.75}{}$ bit per Nyquist interval.
	It can be seen that the proposed solution achieves a marginally lower \gls{tmse} for non-flat input \glspl{psd}.
	This shows the superiority of the proposed pre-sampling filter, which may discard weak frequency components.
	Furthermore, the evaluation shows that the \gls{tmse} is minimized for sampling rates $f_\mathrm{s} < f_\mathrm{nyq}$, for non-flat input \glspl{psd}, which is known from \gls{adx} \cite{8350400,8416707}.\looseness-1
	
	Finally, in \cref{fig:num_res}.(c), we compare the sampling rates, denoted as $f_R$, which minimize the \gls{tmse} for a fixed rate budget $R$ in the proposed system and for the fundamental \gls{adx} limit \cite[eq.~(26)]{8416707}.
	Note that the values of $f_R$ are searched numerically.
	It can be seen that for low rate budgets $R$, the \gls{tmse} is minimized by employing sub-Nyquist sampling.
	Notably, $f_\mathrm{R}$ is much lower in the proposed system compared to the fundamental \gls{adx} limit.
	This demonstrates the effectiveness of sub-Nyquist sampling for practical systems operating under tight rate budgets and employing uniform \glspl{adc}.\looseness-1
	
	\section{Conclusion}\label{sec:conclusion}
	In this work, we studied an acquisition system for \gls{wss} signals employing uniform sampling and uniform quantization.
	For a fixed sampling rate and quantizer resolution, we obtained closed-form expressions for the \gls{tmse}-minimizing pre-sampling and recovery filters, as well as the resulting minimum achievable \gls{tmse}.
	We showed that the proposed solution for the pre-sampling filter is superior to a previously proposed solution for the \gls{pcm} setting from \cite{8416707}.
	Furthermore, our numerical results demonstrated the validity of our model and, most notably, that the \gls{tmse} is often minimized by employing considerable sub-Nyquist sampling for low rate budgets.\looseness-1

	\begin{appendices}
		\renewcommand{\theequation}{A\arabic{equation}}    
		\setcounter{equation}{0}
		
		\section{Proof of Proposition 1}\label{sec:appendix_A}
		\begin{proof}
			Before providing the proof of \cref{prop:opt_recovery_filter}, we introduce the following lemma.\looseness-1
			\begin{lemma}\label{lem:x_e_uncorr}
				The input process $\mathsf{x}(t)$ is uncorrelated to the quantization error $\mathsf{e}[n]$, \iec
				\begin{equation}
					\E\left\lbrace \mathsf{x}(t) \mathsf{e}[n] \right\rbrace = 0.
				\end{equation}
			\end{lemma}
			\begin{proof}
				First, employing the law of total expectation, we obtain
				\begin{equation}
					\E\left\lbrace \mathsf{x}(t) \mathsf{e}[n] \right\rbrace =  \E_{\mathsf{x}}\!\left\lbrace \mathsf{x}(t) \, \E_{\mathsf{e}|\mathsf{x}}\!\left\lbrace  \mathsf{e}[n] \big\vert \mathsf{x}(t) \right\rbrace \right\rbrace.
					\label{eq:lemma1_eq1}
				\end{equation}
				Then, noting that $\mathsf{x}(t) \rightarrow \mathsf{y}[n] \rightarrow \mathsf{e}[n]$ forms a Markov chain \cite[p.~34]{cover1999elements}, where the relationship between $\mathsf{x}(t)$ and $\mathsf{y}[n]$ is deterministic and given by \eqref{eq:def_y_n}, it is sufficient to show that
				$
				\E_{\mathsf{e}|\mathsf{y}}\!\left\lbrace  \mathsf{e}[n] \big\vert \mathsf{y}[n] \right\rbrace = 0,
				\label{eq:lemma1_eq2}
				$
				which holds due to \cite[Th.~2]{gray1993ditheredQuantizers}, hence, concluding the proof.\looseness-1
			\end{proof}
			
			Next we prove \cref{prop:opt_recovery_filter}.
			Note that the following proof is similar to the one provided in \cite[Appendix E]{6484980}.
			In order to minimize \eqref{eq:main_objective}, it is sufficient to minimize $\E \left \lbrace \left \vert \hat{\mathsf{x}}(t) - \mathsf{x}(t) \right \vert^2 \right \rbrace$ for any $t \in \mathbb{R}$, which is feasible here, as we will show in the following.
			From the orthogonality principle it follows \cite[eq.~(7-92)]{papoulis2001probability}
			\begin{IEEEeqnarray}{llCl}
				\label{eq:orth_G_1}
				& \E\left\lbrace \hat{\mathsf{x}}(t) \mathsf{z}[n] \right\rbrace & = &  \E\left\lbrace \mathsf{x}(t) \mathsf{z}[n] \right\rbrace \\
				\Rightarrow \quad & \sum_{k' \in \mathbb{Z}} g(t-k' T_\mathrm{s}) R_{\mathsf{z}}[k'-n] & = &  R_{\mathsf{x} \mathsf{y}}(t-n T_\mathrm{s}),
				\label{eq:orth_G_2} \\
				\Leftrightarrow \quad & \sum_{k \in \mathbb{Z}} g(t - k T_\mathrm{s}) R_{\mathsf{z}}[k] & = & R_{\mathsf{x} \mathsf{y}}(t).
				\label{eq:orth_final}
			\end{IEEEeqnarray}
			with $R_{\mathsf{z}}[l] = \E\left\lbrace \mathsf{z}[n+l] \mathsf{z}[n]\right\rbrace$ and $R_{\mathsf{x} \mathsf{y}}(\tau) = \E\left\lbrace \mathsf{x}(t+\tau) \mathsf{y}(t)\right\rbrace$.
			Above, \eqref{eq:orth_G_2} is obtained using \eqref{eq:def_xHat} and \cref{lem:x_e_uncorr} and \eqref{eq:orth_final} follows by replacing $t$ and $k'-n$ with $t + n T_\mathrm{s}$ and $k$, respectively.
			Noting that the \gls{lhs} of \eqref{eq:orth_final} is equivalent to a convolution of $g(t)$ with $\sum_{k \in \mathbb{Z}} \delta(t-kT_\mathrm{s}) \cdot R_{\mathsf{z}}[k]$, we solve \eqref{eq:orth_final} for $G(f)$ in the Fourier domain, which yields (cf.~\cite[Prop.~3.1]{eldar2015sampling})
			\begin{equation}
				G(f) = S_{\mathsf{x} \mathsf{y}}(f) S^{-1}_{\mathsf{z}}(e^{j 2 \pi f T_\mathrm{s}}),
				\label{eq:G_opt_derived}
			\end{equation}
			with
			\begin{IEEEeqnarray}{lCl}
				\label{eq:def_S_xy}
				S_{\mathsf{x} \mathsf{y}}(f) & = & \FT{R_{\mathsf{x} \mathsf{y}}(t)} = S_{\mathsf{x}}(f) \, H^*(f)\\
				S_{\mathsf{z}}(e^{j 2 \pi f T_\mathrm{s}})  & = & \DTFT{R_{\mathsf{z}}[l]} \nonumber \\
				\label{eq:def_S_z}
				& = & \frac{1}{T_\mathrm{s}} \sum_{k \in \mathbb{Z}} \vert H(f-k f_\mathrm{s}) \vert^2 S_\mathsf{x}(f-k f_\mathrm{s}) + \frac{\Delta^2}{4}. \quad
			\end{IEEEeqnarray}
			Using $\E\{\mathsf{w}^2[n]\} = \frac{\Delta^2}{6}$, $\Delta = \frac{2 \gamma}{2^b}$, and \eqref{eq:def_dynamic_range_squared}, it follows that
			\begin{equation}
				\gamma^2 = \eta^2 \left( \E\left\lbrace \mathsf{y}^2[n] \right\rbrace + \frac{2 \, \gamma^2}{3 \, 2^{2b}}  \right) \Rightarrow \gamma^2 = \bar{\kappa} \, \E\left\lbrace \mathsf{y}^2[n] \right\rbrace,
				\label{eq:gamma2_simplified}
			\end{equation}
			with $\bar{\kappa} = \eta^2 ( 1 - \frac{2 \, \eta^2}{3 \, 2^{2b}} )^{-1}$.
			Again using $\Delta = \frac{2 \gamma}{2^b}$ and utilizing \eqref{eq:gamma2_simplified}, we obtain
			\begin{equation}
				\frac{\Delta^2}{4} = \kappa \int_{-\frac{f_\mathrm{s}}{2}}^{\frac{f_\mathrm{s}}{2}} \sum_{m \in \mathbb{Z}} \vert H(f-m f_\mathrm{s}) \vert^2 S_\mathsf{x}(f-m f_\mathrm{s}) \mathrm{d}f,
				\label{eq:eff_quant_dist}
			\end{equation}
			where $\kappa = \frac{\bar{\kappa}}{2^{2b}}$.
			Finally, \eqref{eq:def_G_opt} is obtained by inserting \eqref{eq:eff_quant_dist} into \eqref{eq:def_S_z} and inserting the result as well as \eqref{eq:def_S_xy} into \eqref{eq:G_opt_derived}.
			
			In order to evaluate the \gls{tmse}, we first evaluate the \emph{instantaneous} \gls{mse}, \iec the \gls{mse} at some time $t \in \mathbb{R}$, which we define as
			\begin{IEEEeqnarray}{lCl}
				\E \left \lbrace \left \vert \hat{\mathsf{x}}(t) - \mathsf{x}(t) \right \vert^2 \right \rbrace & = & \E \left \lbrace \left \vert \mathsf{x}(t) \right \vert^2 \right \rbrace - \E \left \lbrace \hat{\mathsf{x}}(t) \mathsf{x}(t) \right \rbrace \nonumber \\
				& \overset{\mathrm{(a)}}{=} & R_{\mathsf{x}}(0) - \sum_{n \in \mathbb{Z}} g(t - n T_\mathrm{s}) R_{\mathsf{x}\mathsf{y}}(t - n T_\mathrm{s} ) \nonumber \\
				& \overset{\mathrm{(b)}}{=} & R_{\mathsf{x}}(0) - \left( ( g \cdot R_{\mathsf{x}\mathsf{y}} ) * \Sha_{T_\mathrm{s}} \right) (t),
				\label{eq:mse_inst}
			\end{IEEEeqnarray}
			where (a) is due to $R_{\mathsf{x}}(\tau) = \E\left\lbrace \mathsf{x}(t+\tau) \mathsf{x}(t)\right\rbrace$, \eqref{eq:def_z_n}, \eqref{eq:def_xHat}, and \cref{lem:x_e_uncorr}, and (b) is due to writing the sum as a convolution of $g(t) \cdot R_{\mathsf{xy}}(t)$ with $\Sha_{T_\mathrm{s}} = \sum_{n \in \mathbb{Z}} \delta(t - n T_\mathrm{s}) $.
			Employing the inverse \gls{ft}, the second term in \eqref{eq:mse_inst} can be written as
			\begin{IEEEeqnarray}{lCl}
				\left( ( g \cdot R_{\mathsf{x}\mathsf{y}} ) * \Sha_{T_\mathrm{s}} \right) (t) & \overset{\mathrm{(a)}}{=} & \frac{1}{T_\mathrm{s}} \int_{\mathbb{R}} \left( G * S_{\mathsf{x} \mathsf{y}} \right)(f) \cdot \Sha_{f_\mathrm{s}}(f) \, e^{j 2 \pi f t} \mathrm{d}f \nonumber \\
				& \overset{\mathrm{(b)}}{=} & \frac{1}{T_\mathrm{s}} \sum_{n \in \mathbb{Z}} \left( G * S_{\mathsf{x} \mathsf{y}} \right)(n f_\mathrm{s}) \, e^{j 2 \pi n f_\mathrm{s} t}
				\label{eq:help_inst_mse}
			\end{IEEEeqnarray}
			where (a) is due to $\mathrm{FT}\left\lbrace \Sha_{T_\mathrm{s}}(t) \right\rbrace = \frac{1}{T_\mathrm{s}}\Sha_{f_\mathrm{s}}(f)$ and (b) is due to inserting the definition of $\Sha_{f_\mathrm{s}}(f)$ and exchanging the integral and sum operations. Then, averaging \eqref{eq:help_inst_mse} over one sampling period, \iec averaging over $t \in [0, T_\mathrm{s})$, yields
			\begin{IEEEeqnarray}{lCl}
				\frac{1}{T_\mathrm{s}} \int_0^{T_\mathrm{s}} \left( ( g \cdot R_{\mathsf{x}\mathsf{y}} ) * \Sha_{T_\mathrm{s}} \right) (t) \mathrm{d} t & \overset{\mathrm{(a)}}{=} & \frac{1}{T_\mathrm{s}} \left( G * S_{\mathsf{x} \mathsf{y}} \right)(0) \nonumber\\
				& \overset{\mathrm{(b)}}{=} & \frac{1}{T_\mathrm{s}} \int_{\mathbb{R}} G(f) S_{\mathsf{y} \mathsf{x}} (f) \mathrm{d}f, \qquad
				\label{eq:partial_inst_mse_time_avg}
			\end{IEEEeqnarray}
			where (a) is due to inserting \eqref{eq:help_inst_mse} and solving the integral and (b) is due to $ S_{\mathsf{x} \mathsf{y}} (-f) = \mathrm{FT}\left\lbrace R_{\mathsf{x} \mathsf{y}} (-t) \right\rbrace = \mathrm{FT}\left\lbrace R_{\mathsf{y} \mathsf{x}} (t) \right\rbrace = S_{\mathsf{y} \mathsf{x}} (f)$.
			Finally, substituting \eqref{eq:help_inst_mse} into \eqref{eq:mse_inst} and inserting the result into the objective of \eqref{eq:main_objective} and utilizing \eqref{eq:G_opt_derived}, \eqref{eq:def_S_xy}, and \eqref{eq:partial_inst_mse_time_avg} proves \eqref{eq:mse_prop_1}.
		\end{proof}
		
		\section{Proof of Theorem 1}\label{sec:appendix_B}
		\begin{proof}
			First we note that the \gls{tmse} expression from \cref{prop:opt_recovery_filter}, given in \eqref{eq:mse_prop_1}, is minimized by maximizing the second term.
			Hence, we obtain the following constrained maximization problem:
			\begin{subequations}
				\label{eq:mse_obj_full_constrained}
				\begin{IEEEeqnarray}{ll}
					\label{eq:mse_obj_full_constrained_obj}
					\max_{\bar{H}(f) \geq 0} \quad & \int_{-\frac{f_\mathrm{s}}{2}}^{\frac{f_\mathrm{s}}{2}} \frac{ \sum_{k \in \mathbb{Z}} \bar{H}(f - k f_\mathrm{s}) S_{\mathsf{x}}(f - k f_\mathrm{s}) }{ \sum_{k' \in \mathbb{Z}} \bar{H}(f - k' f_\mathrm{s}) + \frac{1}{2^{2b}} } \mathrm{d}f \\
					\label{eq:mse_obj_full_constrained_const}
					\mathrm{~~s.t.} \quad & \bar{\kappa} T_\mathrm{s} \int_{-\frac{f_\mathrm{s}}{2}}^{\frac{f_\mathrm{s}}{2}} \sum_{k \in \mathbb{Z}} \bar{H}(f' - k f_\mathrm{s}) \mathrm{d}f' = 1,
				\end{IEEEeqnarray}
			\end{subequations}
			with
			$\bar{H}(f) = \left\vert H(f) \right\vert^2 S_{\mathsf{x}}(f)$
			and where the objective, given in \eqref{eq:mse_obj_full_constrained_obj}, is obtained by writing the integral as a sum over integrals of support $f_\mathrm{s}$ and subsequently exchanging the sum and integral operations.
			\begin{lemma}
				\label{lem:at_most_one_aliased_component}
				The objective \eqref{eq:mse_obj_full_constrained_obj} is maximized by setting for each $|\tilde{f}_0| < \frac{f_\mathrm{s}}{2}$ at most a single value of $\{ \bar{H}(\tilde{f}_0 + k f_\mathrm{s}) \}_\mathrm{k \in \mathbb{N}}$ to be non-zero and this value corresponds to the aliased frequency, \iec to the $k \in \mathbb{N}$, with the maximal value of $\{ S_{\mathsf{x}}(\tilde{f}_0 + k f_\mathrm{s}) \}_\mathrm{k \in \mathbb{N}}$.
			\end{lemma}
			\begin{proof}
				To prove the lemma, we show that each filter, which does not satisfy the conditions of the lemma can be improved upon by an alternative filter, which still satisfies the constraints but yields a higher objective value.
				To this aim, assume there exists a filter $\{ \bar{H}^\mathrm{c}(f) \}_{f \in \mathbb{R}}$ which satisfies the constraints $\bar{H}(f) \geq 0$ and \eqref{eq:mse_obj_full_constrained_const}.
				Furthermore, assume for some $f_0 \in \mathbb{R}$ this filter satisfies $\bar{H}^\mathrm{c}(f_0) > 0$ and $\bar{H}^\mathrm{c}(f_0 + k f_\mathrm{s}) > 0$ for some $k \in \mathbb{N}$, while it holds $S_{\mathsf{x}}(f_0) \geq S_{\mathsf{x}}(f_0 + k f_\mathrm{s})$.\looseness-1
				
				Now, we construct an alternative filter $\{ \bar{H}^\mathrm{a}(f) \}_{f \in \mathbb{R}}$ such that\looseness-1
				\begin{equation}
					\label{eq:def_H_alt}
					\bar{H}^\mathrm{a}(f) = \begin{cases}
						\bar{H}^\mathrm{c}(f_0) + \bar{H}^\mathrm{c}(f_0 + k f_\mathrm{s}), & \text{for~} f = f_\mathrm{0}\\
						0, & \text{for~} f = f_0 + k f_\mathrm{s}\\
						\bar{H}^\mathrm{c}(f), & \text{otherwise.}
					\end{cases}
				\end{equation}
				We note that $\sum_{k \in \mathbb{N}} \bar{H}^\mathrm{a}(f - k f_\mathrm{s}) = \sum_{k \in \mathbb{N}} \bar{H}^\mathrm{c}(f - k f_\mathrm{s})$.
				Hence, it can be verified that also $\bar{H}^\mathrm{a}(f)$, defined in \eqref{eq:def_H_alt}, satisfies $\bar{H}(f) \geq 0$ and \eqref{eq:mse_obj_full_constrained_const}.
				Then, it can be shown that the integrand of the objective \eqref{eq:mse_obj_full_constrained_obj} at the frequency $\tilde{f}_0 \in \left(-\frac{f_\mathrm{s}}{2},\frac{f_\mathrm{s}}{2}\right)$ to which $f_0$ is aliased satisfies
				\begin{equation*}
					\frac{ \sum\limits_{k \in \mathbb{Z}}\!\bar{H}^\mathrm{a}(\tilde{f}_0 - k f_\mathrm{s})S_{\mathsf{x}}(\tilde{f}_0 - k f_\mathrm{s})\!\! }{ \sum\limits_{k \in \mathbb{Z}}\! \bar{H}^\mathrm{a}(\tilde{f}_0 - k f_\mathrm{s}) + \frac{1}{2^{2b}} } \!\geq\! \frac{ \sum\limits_{k \in \mathbb{Z}}\!\bar{H}^\mathrm{c}(\tilde{f}_0 - k f_\mathrm{s})S_{\mathsf{x}}(\tilde{f}_0 - k f_\mathrm{s})\!\! }{ \sum\limits_{k \in \mathbb{Z}}\! \bar{H}^\mathrm{c}(\tilde{f}_0 - k f_\mathrm{s}) + \frac{1}{2^{2b}} }.
				\end{equation*}
				Consequently, employing $\{ \bar{H}^\mathrm{a}(f) \}_{f \in \mathbb{R}}$ instead of $\{ \bar{H}^\mathrm{c}(f) \}_{f \in \mathbb{R}}$ can only result in an improved objective, which concludes the proof.
			\end{proof}
			
			From \cref{lem:at_most_one_aliased_component} it follows that the maximization problem given in \eqref{eq:mse_obj_full_constrained} can equivalently be written as
			\begin{subequations}
				\label{eq:mse_obj_full_constrained_lemma}
				\begin{IEEEeqnarray}{ll}
					\label{eq:mse_obj_full_constrained_obj_lemma}
					\max_{\tilde{H}(f) \geq 0} \quad & \int_{-\frac{f_\mathrm{s}}{2}}^{\frac{f_\mathrm{s}}{2}} \frac{ \tilde{H}(f) \tilde{S}_{\mathsf{x}}(f) }{ \tilde{H}(f) + \frac{1}{2^{2b}} } \mathrm{d}f \\
					\label{eq:mse_obj_full_constrained_const_lemma}
					\mathrm{~~s.t.}& \bar{\kappa} T_\mathrm{s} \int_{-\frac{f_\mathrm{s}}{2}}^{\frac{f_\mathrm{s}}{2}} \tilde{H}(f') \mathrm{d}f' = 1,
				\end{IEEEeqnarray}
			\end{subequations}
			where we define $\tilde{k}(f) = \argmax_{k \in \mathbb{Z}} S_{\mathsf{x}}(f - k f_\mathrm{s})$, such that $\tilde{S}_{\mathsf{x}}(f) = S_{\mathsf{x}}(f - \tilde{k}(f) f_\mathrm{s}) \, \mathds{1}_{|f|<\frac{f_\mathrm{s}}{2}}(f)$ and $\tilde{H}(f) = \bar{H}(f - \tilde{k}(f) f_\mathrm{s}) \, \mathds{1}_{|f|<\frac{f_\mathrm{s}}{2}}(f)$.
			
			Next, we note that the integrand of the objective \eqref{eq:mse_obj_full_constrained_obj_lemma} is concave in $\tilde{H}(f)$ for $\tilde{H}(f) \geq 0$.
			Hence, the objective \eqref{eq:mse_obj_full_constrained_obj_lemma} is also concave in $\tilde{H}(f)$ for $\tilde{H}(f) \geq 0$, as the integral operation preserves concavity \cite[Sec.~3.2.1]{boyd2004convex}.
			Furthermore, because the objective and the constraints of \eqref{eq:mse_obj_full_constrained_lemma} are differentiable and because the equality constraint \eqref{eq:mse_obj_full_constrained_const_lemma} is affine in $\tilde{H}(f)$, the \gls{kkt} conditions are necessary and sufficient for optimality \cite[Sec.~5.5.3]{boyd2004convex}.
			It can be shown that
			$\tilde{H}(f) = \frac{1}{2^{2b}} \Big( \sqrt{ \zeta \, \tilde{S}_{\mathsf{x}}(f) } - 1 \Big)^+$,
			satisfies these \gls{kkt} conditions, where $\zeta$ has to be chosen such that \eqref{eq:mse_obj_full_constrained_const_lemma} holds (cf.~\cite[Appendix~B]{neuhaus2021task}).
			This proves \eqref{eq:theorem_def_H_tilde_opt}.
			Finally, \eqref{eq:theorem_def_TMSE_opt} is obtained by inserting the objective of \eqref{eq:mse_obj_full_constrained_obj_lemma} into \eqref{eq:mse_prop_1}, which concludes the proof.
		\end{proof}
	\end{appendices}
	
	\bibliographystyle{IEEEtran}
	\bibliography{ref}
	
\end{document}